\documentclass[11pt]{article}
\usepackage[T1]{fontenc}
\usepackage{kpfonts}
\usepackage[dvipsnames]{xcolor}
\usepackage{pdfpages}
\usepackage{sgame}
\usepackage{accents}
\usepackage{tikz-cd}
\usepackage{float}
\usepackage[round]{natbib}   
\usepackage{multirow}
\usepackage{dutchcal}
\usepackage{pdfpages}
\usepackage{bbm}
\usepackage{sgame}
\usepackage{mathtools}
\usepackage{amsthm}
\usepackage{enumitem}
\usepackage[margin=.9in]{geometry}
\usepackage{caption}
\usepackage{subcaption}
\usepackage{epigraph}
\usepackage{listings}
\usetikzlibrary{calc}
\usetikzlibrary{shapes,arrows}

\newtheorem{theorem}{Theorem}

\theoremstyle{definition}

\usepackage{hyperref}
\hypersetup{
    colorlinks=true,
    linkcolor=CadetBlue,
    filecolor=CadetBlue,      
    urlcolor=CadetBlue,
    citecolor=CadetBlue,
}

\definecolor{backcolour}{rgb}{0.63, 0.79, 0.95}

\lstdefinestyle{mystyle}{
  backgroundcolor=\color{backcolour},
  basicstyle=\ttfamily\footnotesize,
  breakatwhitespace=false,         
  breaklines=true,                 
  captionpos=b,                    
  keepspaces=true,                 
  numbers=left,                    
  numbersep=5pt,                  
  showspaces=false,                
  showstringspaces=false,
  showtabs=false,                  
  tabsize=2
}


\begin{document}
\author{Mark Whitmeyer\thanks{Arizona State University. \href{mailto:mark.whitmeyer@gmail.com}{mark.whitmeyer@gmail.com}. Andy Kleiner, Joseph Whitmeyer, and Kun Zhang gave me useful feedback, like, e.g., ``give the reader a more detailed proof.''}}

\title{A More Informed Sender Benefits the Receiver When the Sender Has Transparent Motives}
\date{}
\maketitle

A sender (\(S\)) with state-independent preferences (i.e., \textbf{transparent motives}, \cite{transmotives}) privately observes a signal, \(\pi\), about the state of the world, before sending a message to a receiver (\(R\)), who subsequently takes an action. Regardless of whether \(R\) can \textbf{mediate}--and commit to a garbling of \(S\)'s message--or \textbf{delegate}--commit to a stochastic decision rule as a function of \(S\)'s message--and understanding the statement ``\(R\) is better off as a result of an improvement of \(S\)'s information from \(\pi\) to \(\pi'\),'' where \(\pi\) is a garbling of \(\pi'\), to mean that her maximal and minimal equilibrium payoffs (weakly) increase,
\begin{theorem}
If \(S\) is more informed, \(R\) is better off.
\end{theorem}
\begin{proof}
Regardless of whether \(R\) can delegate or mediate, \(S\)'s strategy induces an arbitrary garbling of his signal; and as transparent motives mean that \(S\) must be indifferent over all on-path messages, she is, therefore, willing to mix (at equilibrium) in a way that produces garbling \(\rho\), where \(\pi = \rho \circ \pi' \). \end{proof}
This result is related to the remarkable finding of \cite{lichtig} that \(R\)'s payoff in the truth-leaning equilibrium in a class of games with hard evidence and transparent motives improves as \(S\)'s becomes more informed. This positive relationship (between \(S\)'s information and \(R\)'s welfare) is not generally present in cheap-talk and signaling games; nor must transparently motivated \(S\)'s maximal payoff improve as he becomes more informed.

\bibliography{sample.bib}

\end{document}